\documentclass[journal,twoside]{IEEEtran}
\usepackage{graphicx}
\usepackage{url}
\usepackage{amsfonts}
\usepackage{amsmath}
\usepackage{amsmath}
\usepackage{amsthm}
\usepackage{amssymb}
\usepackage{mathtools}
\usepackage{authblk}
\usepackage[top=1.25in, bottom=1.5in, left=1.5in, right=1.5in]{geometry}
\usepackage{xcolor}
\usepackage{tikz}
\usepackage{lmodern}

\newtheorem{Theorem}{Theorem}

\newtheorem{Example}[Theorem]{Example}
\newtheorem{Remark}[Theorem]{Remark}
\newtheorem{Lemma}[Theorem]{Lemma}
\newtheorem{Corollary}[Theorem]{Corollary}	
	
\newtheorem{Construction}[Theorem]{Construction}

\newcommand{\SSS}{{\ensuremath{\mathcal{S}}}} 
\newcommand{\PPP}{{\ensuremath{\mathcal{P}}}}

\newcommand{\YYY}{{\ensuremath{\mathcal{Y}}}} 
\newcommand{\ZZZ}{{\ensuremath{\mathcal{Z}}}} 
\newcommand{\XXX}{{\ensuremath{\mathcal{X}}}} 

\newcommand{\Pro}{{\ensuremath{\mathsf{Pr}}}}

\begin{document}
\title{Constructions and bounds for codes with restricted overlaps}

\author{%
Simon~R.~Blackburn,~\IEEEmembership{Senior Member,~IEEE,}
Navid~Nasr~Esfahani,~\IEEEmembership{Member,~IEEE,}
Donald~L.~Kreher,
and~Douglas~R.~Stinson
\thanks{S.~R.~Blackburn is with the
Department of Mathematics, Royal Holloway, 
University of London, Egham, Surrey TW20 0EX,United Kingdom}
\thanks{N.~N.~Esfahani is with the
Department of Computer Science, Memorial University of Newfoundland, St. John's, NL A1B 3X5, Canada}
\thanks{D.~L.~Kreher is with the
Department of Mathematical Sciences,
Michigan Technological University,
Houghton, MI 49931-1295, U.S.A.}
\thanks{D.~R.~Stinson is with the
David R.~Cheriton School of Computer Science, University of Waterloo,  Waterloo ON, N2L 3G1, Canada}
\thanks{D.~R.~Stinson is also with the 
School of Mathematics and Statistics,
Carleton University,
Ottawa, Ontario, K1S 5B6, Canada}
\thanks{D.~R.~Stinson's research is supported by  NSERC discovery grant RGPIN-03882.}}%


\markboth{IEEE Transactions on Information Theory}%
{Blackburn \MakeLowercase{\textit{et al.}}:Constructions and bounds for codes with restricted overlaps}

\maketitle

\begin{abstract}
Non-overlapping codes have been studied for almost 60 years. In such a code, no proper, non-empty prefix of any codeword is a suffix of any codeword.
In this paper, we study codes in which overlaps of certain specified sizes  are forbidden. We prove some general bounds and we give several constructions in the case of binary codes. Our techniques also allow us to provide an alternative, elementary proof of a lower bound on non-overlapping codes due to Levenshtein \cite{Lev64} in 1964.
\end{abstract}


\section{Introduction}

Let $u$ and $v$ be (not necessarily distinct) words of length $n$ over a specified alphabet. Let $t$ be an integer such that $1 \leq t \leq n-1$. We say that $u$ and $v$ have a \emph{$t$-overlap} if the prefix of $u$ of length $t$ is identical to the suffix of $v$ of length $t$. A code $C$ is \emph{$t$-overlap-free} if no codewords $u$ and $v$ in $C$ have a $t$-overlap. A code $C$ is \emph{non-overlapping} if it is $t$-overlap-free for all $t$ such that $1 \leq t \leq n-1$. 

Motivated by applications including frame synchronization, non-overlapping codes have been studied by numerous authors over the years, e.g., see \cite{Bil,BPP12,Bl,Chee,Lev64,Lev70,Wang}. 

Here we consider a less restrictive definition. Suppose that $t_1$ and $t_2$ are integers such that $1 \leq t_1 \leq t_2 \leq n-1$. We say that a code $C$ is \emph{$(t_1,t_2)$-overlap-free} if it is $t$-overlap-free for all $t$ such that $t_1 \leq t \leq t_2$. Two special cases of interest are
codes that are $(k,n-1)$-overlap-free (i.e., overlaps of size at least $k$ are not allowed) and codes that are $(1,k)$-overlap-free (i.e., overlaps of size at most $k$ are not allowed).  

Motivated by applications in DNA-based storage systems and synchronization protocols, 
$(k, n - 1)$-overlap-free codes were studied in \cite{YKM} and
termed $k$-weakly mutually uncorrelated codes. On the other hand, $(1,k)$-overlap-free  codes could be useful in a setting where we  have ``approximate'' synchronization, i.e., if we can assume that codewords will not ``drift'' too much. For example, suppose (see Figure~\ref{fig:application}) that we transmit blocks $1$, $2$, $3$ and so on, each a codeword of the same length $n$. We consider channels where a received block might be corrupted, with bits changed and up to $k$ bits inserted or deleted. We detect a loss of synchronization by checking if each block of $n$ received bits is a codeword. If we use an $(n-k,n-1)$-overlap-free code, we are guaranteed to detect a loss of synchronization after $2n$ bits are received. If we use a $(1,k)$-overlap-free code, we are guaranteed to detect a loss of synchronization after $3n$ bits if there are inserted bits, but after only $n$ bits are received if bits have been deleted. Thus, in channels where deletions are more likely than insertions, $(1,k)$-overlap-free codes have an advantage over $(n-k,n-1)$-overlap-free codes.
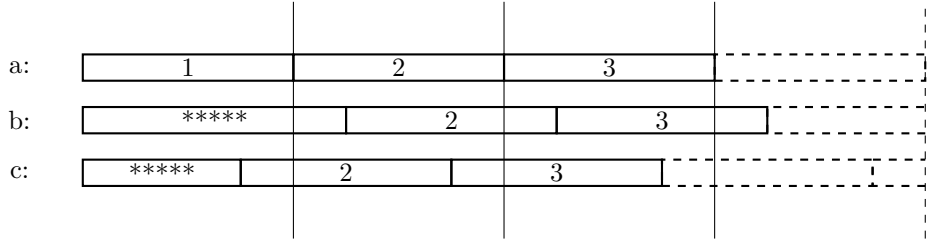
\begin{figure*}[t]
\begin{center}
\begin{tikzpicture}[scale=0.7]
\draw (-1.2,0.25) node {a:};
\draw[thick] (0,0) rectangle (4,0.5) node[midway] {1};
\draw[thick] (4,0) rectangle (8,0.5) node[midway] {2};
\draw[thick] (8,0) rectangle (12,0.5) node[midway] {3};
\draw[dashed,thick] (12,0) rectangle (16,0.5);

\draw (-1.2,-0.75) node {b:};
\draw[thick] (0,-1) rectangle (5,-0.5) node[midway] {*****};
\draw[thick] (5,-1) rectangle (9,-0.5) node[midway] {2};
\draw[thick] (9,-1) rectangle (13,-0.5) node[midway] {3};
\draw[dashed,thick] (16,-1) -- (13,-1) -- (13,-0.5) -- (16,-0.5);

\draw (-1.2,-1.75) node {c:};
\draw[thick] (0,-2) rectangle (3,-1.5) node[midway] {*****};
\draw[thick] (3,-2) rectangle (7,-1.5) node[midway] {2};
\draw[thick] (7,-2) rectangle (11,-1.5) node[midway] {3};
\draw[dashed,thick] (11,-2) rectangle (15,-1.5);
\draw[dashed,thick] (16,-2) -- (15,-2) -- (15,-1.5) -- (16,-1.5);

\draw (4,-3) -- (4,1.5);
\draw (8,-3) -- (8,1.5);
\draw (12,-3) -- (12,1.5);
\draw[dashed] (16,-3) -- (16,1.5);
\end{tikzpicture}
 \end{center}
\caption{Transmitting codewords when blocks are corrupted, and change length: (a) Transmitted data, (b) Inserted and corrupted bits, and (c) Deleted and corrupted bits.}
\label{fig:application}
 \end{figure*}
 
We comment that codes for synchronization is a large and thriving area, which we cannot hope to cover comprehensively here. Notable related problems are codes designed to correct bursts of insertions or deletions~\cite{Cheng, Lev65, Lev70a, Schoe}, and variants of the non-overlapping problem in two-dimensions~\cite{Bar}.

In general, we wish to determine the maximum number of codewords in a $(t_1,t_2)$-overlap-free code. 
In Section \ref{twobounds.sec}, we prove two upper bounds, on the size of $(k,n-1)$-overlap-free codes and $(1,k)$-overlap-free codes. Section \ref{construction.sec} begins a study of constructions for $(1,k)$-overlap-free codes over a binary alphabet. Our first construction, the \textsf{Doubling Construction}, gives an inductive approach to the construction of these codes. 
Section \ref{graph.sec} introduces a graph-based interpretation of these codes. This approach is used to prove the optimality of our codes for $k \leq 6$. Section \ref{m-min.sec} presents an explicit construction that we term the \textsf{$m$-minimum Construction}, as well as the closely related \textsf{Zero Block Construction}. Both of these permit ``good'' codes to be constructed for specified values of $k$. The second of these two constructions can be analyzed by exploiting a connection with $n$-step Fibonacci numbers. 
We provide exact as well as asymptotic bounds; it is shown that the constructed codes are within a small constant factor of being optimal. Section \ref{classical.sec} revisits the classical problem of non-overlapping codes and discusses how our techniques apply to this problem. In particular, we provide an alternative, elementary proof of a lower bound on non-overlapping codes due to Levenshtein \cite{Lev64} in 1964. Finally, Section \ref{summary.sec} is a brief discussion and summary.

\section{Two upper bounds}
\label{twobounds.sec}

Chee \emph{et al.}~\cite{Chee} proved that if $C$ is a non-overlapping code over an alphabet of cardinality $q$, then $|C| \leq q^n / (2n-1)$. This bound can be proven using a simple combinatorial argument; see Blackburn \cite{Bl}. Also, a stronger bound has been proven by Levenshtein \cite{Lev70} using analytic combinatorics.

For $(k,n-1)$-overlap-free codes, Yazdi \emph{et al.} \cite{YKM} proved that such a code $C$ satisfies the inequality $|C| \leq q^n / (n-k+1)$. The following  stronger bound can be proven using the argument from \cite{Bl}. Note that the special case $k=1$ of Theorem \ref{lower.thm} is essentially the bound proven in \cite{Bl}.

\begin{Theorem}
\label{lower.thm}
If $C$ is a $(k,n-1)$-overlap-free code over an alphabet of cardinality~$q$, then
\[ |C| \leq \frac{q^n}{2n-2k+1}.\]
\end{Theorem}

\begin{proof}
Let $C$ be a $(k,n-1)$-overlap-free code over an alphabet $F$ of cardinality $q$. 
For $ w \in F^{2n-2k+1}$ and $1\leq i \leq 2n-2k+1$, define 
$w(i) = (w_i, w_{i+1}, \dots , w_{i + n-1})$, where the subscripts are reduced modulo $2n-2k+1$.
Thus, $w(i)$ is the cyclic subword of length $n$ of $w$ starting at $w_i$. 
Define 
\begin{multline*}
X = \{ (w,i): w \in F^{2n-2k+1},\\
1\leq i \leq 2n-2k+1, w(i) \in C \} .
\end{multline*}

Suppose there exists $w\in F^{2n-2k+1}$ and $i, i'$ such that $i \neq i'$ and $(w,i), (w,i') \in X$. We claim that 
$(w,i)$ and $(w,i')$ have an overlap of size at least $k$. This occurs because the overlap between $(w,i)$ and $(w,i')$ is 
at least \[n + n - (2n - 2k+1) = 2k-1,\] and hence the overlap at one end is at least $\lceil (2k-1)/2 \rceil = k$.  This violates the non-overlapping properties of $C$. Hence, for each $w\in F^{2n-2k+1}$ there is at most one $i$ such that $(w,i) \in X$. Thus it follows that \[|X| \leq q^{2n-2k+1}.\]

Also, $|X| = (2n-2k+1)|C|q^{n-2k+1}$, since there are $2n-2k+1$ choices for $i$, $|C|$ choices for $w(i)$, and 
$q^{n-2k+1}$ choices for the remaining entries in $w$.

Hence, 
\[(2n-2k+1)|C|q^{n-2k+1} \leq q^{2n-2k+1},\]
which immediately yields the stated upper bound on $|C|$.
\end{proof}

It is natural to ask if there is a ``related'' upper bound for $(1,k)$-overlap-free codes.

\begin{Theorem}
\label{1k.thm}
Let $C$ be a $(1,k)$-overlap-free code, where $k\leq n/2$. Then
\[
|C|\leq 
\tfrac{1}{2k}q^n.
\]
\end{Theorem}

\begin{proof}
Let $k\leq n/2$ and let $C$ be a $(1,k)$-overlap-free code. Let $\XXX$ be the set of all codewords with the middle $n-2k$ positions removed. Clearly, $|C|\leq q^{n-2k}|\XXX |$. The elements of $\XXX$ are $q$-ary words of length $2k$. We have $\XXX = \YYY \cup \ZZZ$, where the elements in $\YYY$ have (cyclic) period strictly dividing $2k$, and the elements of $\ZZZ$ have period exactly $2k$. 

Suppose $y\in{\cal Y}$. If $y$ has period $p$, where $p$ strictly divides $2k$, then $p\leq k$. Then the first and last $p$ elements of a corresponding codeword $w\in{\cal X}$ agree. This codeword has a $p$-overlap with itself, which contradicts the $(1,k)$-overlap-free property. We conclude that ${\cal Y} = \emptyset$.

Now we claim that no pair of distinct elements in $\ZZZ$ are cyclic shifts of each other. For a contradiction, suppose that $z_1,z_2$ are a pair of distinct elements from $\ZZZ$ that are cyclic shifts of each other. Let $c_1,c_2$ be the corresponding codewords in $C$. Write $\sigma$ for the `cyclic shift left by one position' operator, so \[\sigma(a_1a_2\cdots a_{2k})=a_2a_3\cdots a_{2k}a_1.\] 
Then $z_1{=}\sigma^j(z_2)$ for some 
$j{\in}\{1,\ldots,2k-1\}$.
Swapping $z_1$ and $z_2$ if needed, we may assume that $j$ lies in the set $\{1,\ldots,k\}$ (as swapping replaces $j$ by $2k-j$). But now the $j$-prefix of $z_2$ is equal to the $j$-suffix of $z_1$. So the $j$-prefix of $c_2$ is equal to the $j$-suffix of $c_1$. This contradicts our assumption that $C$ is $(1,k)$-overlap-free, and so our claim follows.

We can partition the set of all $q$-ary sequences of length $2k$ and period exactly $2k$ into equivalence classes under cyclic shift. Each class contains $2k$ sequences, and so there are at most $q^{2k}/2k$ classes. The previous paragraph shows that no class contains two elements of $\ZZZ$, and so $|\ZZZ |\leq q^{2k}/2k$. Hence
\begin{align*}
|C|\leq q^{n-2k}|\XXX |&=q^{n-2k}(|\YYY |+|\ZZZ |)\\
&=q^{n-2k}|\ZZZ |\\
&\leq \frac{q^n}{2k}. \qedhere
\end{align*}
\end{proof}

\section{Constructions}
\label{construction.sec}

In this section, and the next two sections, we investigate constructions and bounds for $(1,k)$-overlap-free codes. All of our constructions will be based on the following template.

\begin{Construction}
Let $F$ be an alphabet of size $q$ and let $n$ and $t$ be positive integers such that $n \geq 2t$.
Let $\PPP$ and $\SSS$ be two sets of $t$-tuples from $F^t$. Define 
\begin{multline*}
C(\PPP,\SSS,n,t) =
 \{ p \parallel x \parallel s: p \in \PPP,\\
s \in \SSS, x \in F^{n-2t} \} .
\end{multline*}
Thus a codeword $c \in C$ has a prefix chosen from $\PPP$, a suffix chosen from $\SSS$, and the remaining $n-2t$ elements are arbitrary symbols from $F$. We also observe that $|C(\PPP,\SSS,n,t)| = |\PPP| \times |\SSS| \times q^{n-2t}$.
\end{Construction}

The following Lemma is immediate.

\begin{Lemma}
\label{disjoint.lem}
$C(\PPP,\SSS,n,t)$  is $t$-overlap-free if and only  if $\PPP  \cap \SSS = \emptyset$.
\end{Lemma}

Suppose $\PPP$ and $\SSS$ are two sets of $k$-tuples from $F^k$. Suppose $t$ is a positive integer such that $t < k$. Define
$\PPP|_t$ to be the set of all $t$-prefixes of tuples from $\PPP$ and 
$\SSS|_t$ to be the set of all $t$-suffixes of tuples from $\SSS$. So
\begin{multline*}
\PPP|_t = \{ (p_1,\dots , p_t) : \text{ there exists }\\
(p_1,\dots ,p_t,p_{t+1},\ldots, p_k) \in \PPP\} ,
\end{multline*}
and
\begin{multline*}
\SSS|_t = \{ (s_{k-t+1},\dots , s_k) : \text{ there exists }\\
(s_1,\dots,s_{k-t},s_{k-t+1},\ldots , s_k) \in \SSS\} .
\end{multline*}

The following is a straightforward extension of Lemma \ref{disjoint.lem}.

\begin{Theorem}
\label{disjoint.thm}
$C(\PPP,\SSS,n,t)$  is a $(1,k)$-overlap-free code if and only if $\PPP|_t  \cap \SSS|_t = \emptyset$ for $1 \leq t \leq k$.
\end{Theorem}

\subsection{The \textsf{Doubling Construction}}
\label{sec3}

Theorem \ref{disjoint.thm} suggests a way to build up $(1,k)$-overlap-free codes inductively. We will refer to this process as the \textsf{Doubling Construction}. For the rest of the paper, we consider the binary case, where $q = 2$. 

We take $F = \{0,1\}$. Suppose we begin with $k=1$. Without loss of generality, we can define $\PPP = \{0\}$ and $\SSS = \{1\}$. So $C(\PPP,\SSS,n,1)$ would consist of all $2^{n-2}$ binary $n$-tuples that begin with a $0$ and end with a $1$.

Next, we consider $k=2$. We consider extensions of the solution for $k=1$, where we append a symbol to a tuple in $\PPP$ and we prepend a symbol to a tuple in $\SSS$:

\[
\begin{array}{c|c}
\PPP & \SSS \\\hline
0 0 & \textcolor{red}{0 1} \\
\textcolor{red}{0 1} & 1 1
\end{array}
\]

We cannot include $01$ in both  $\PPP $ and $\SSS $. Without loss of generality, we include $01$ in $\PPP $ but not in  $\SSS $.
So we obtain the following solution for $k=2$: 
\begin{center}
$\PPP = \{00,01\}$ and $\SSS = \{11\}$.
\end{center} Thus $C(\PPP,\SSS,n,2)$ 
would consist of
\[2 \times 2^{n-4}\]
binary $n$-tuples.

We can use a similar process to proceed from $k=2$ to $k=3$.
We append a symbol to each tuple in $\PPP$ and we prepend a symbol to each tuple in $\SSS$:
\[
\begin{array}{c|c}
\PPP & \SSS \\\hline
0 0 0 & \textcolor{red}{0 1 1} \\
0 0 1 & 1 1 1\\
0 1 0 & \\
\textcolor{red}{0 1 1} &
\end{array}
\]
Now the $3$-tuple $011$ is duplicated. We retain it in $\SSS$ and delete it from $\PPP$ (this will lead to the largest code, since $3 \times 2 > 4 \times 1$.
We obtain the following solution for $k=3$: $\PPP = \{000,001,010\}$ and $\SSS = \{011,111\}$. Thus $C(\PPP,\SSS,n,3)$  consists of  
\[3 \times 2 \times 2^{n-6} = 6 \times 2^{n-6}\] binary $n$-tuples.

Now we proceed from $k=3$ to $k=4$.
We get the following:
\[
\begin{array}{c|c}
\PPP & \SSS \\\hline
0 0 0 0 & \textcolor{red}{0 0 1 1} \\
0 0 0 1 & {1 0 1 1} \\
0 0 1 0 & 0 1 1 1\\
\textcolor{red}{0 0 1 1} & 1 1 1 1\\
0 1 0 0 & \\ 
0 1 0 1 & \\ 
\end{array}
\]
The $4$-tuple $0011$ is duplicated. Again, we retain it in $\SSS$ and delete it from $\PPP$. 
We obtain the following solution for $k=4$: 
\begin{align*}
\PPP &= \{0000,0001,0010,0100,0101\}\\
\shortintertext{and}
 \SSS &= \{0011,1011,0111,1111\}.
\end{align*}
Thus $C(\PPP,\SSS,n,4)$  consists of  
\[
5 \times 4 \times 2^{n-8} = 20 \times 2^{n-8}\
\]
binary $n$-tuples.

When we proceed from $k=4$ to $k=5$,
we obtain the following:
\[
\begin{array}{c|c}
\PPP & \SSS \\\hline
0 0 0 0 0 & \textcolor{red}{0 0 0 1 1} \\
0 0 0 0 1 & 1 0 0 1 1 \\
0 0 0 1 0 & \textcolor{blue}{0 1 0 1 1} \\
\textcolor{red}{0 0 0 1 1} & 1 1 0 1 1 \\
0 0 1 0 0 & 0 0 1 1 1\\
0 0 1 0 1 & 1 0 1 1 1\\
0 1 0 0 0 & 0 1 1 1 1\\
0 1 0 0 1 & 1 1 1 1 1\\
0 1 0 1 0 & \\ 
\textcolor{blue}{0 1 0 1 1} & \\ 
\end{array}
\]
Now there are two duplicated $5$-tuples. We will retain both $5$-tuples in $\SSS$ in order to balance the sizes of $\PPP$ and $\SSS$. 
So we obtain the following solution for $k=5$:  
\[\PPP = \left\{
\begin{array}{@{}l@{}}
00000,00001,00010,00100,\\00101,01000,01001,01010
\end{array}
\right\}\] and 
\[\SSS = \left\{
\begin{array}{@{}l@{}}
00011,10011,01011,11011,\\00111,10111,01111,11111
\end{array}
\right\}.\] 
Thus $C(\PPP,\SSS,n,5)$  consists of  
\[8 \times 8 \times 2^{n-10} = 2^{n-4}\] binary $n$-tuples.

We can make a few observations as to what happens when we increase $k$ by one in the \textsf{Doubling Construction}.
\begin{enumerate}
\item First, we double the size of $\PPP$ and $\SSS$ by appending $0$ and $1$ to every tuple in $\PPP$ and prepending $0$ and $1$ to every tuple in $\SSS$.
\item Then we look for duplicates in $\PPP$ and $\SSS$. Note that a duplicate occurs in the new $\PPP$ and $\SSS$ whenever there was a $k$-tuple in the old $\PPP$ whose suffix of size $k-1$ is identical to a prefix of size $k-1$ of a $k$-tuple in the old $\SSS$. For example, when $k=4$, we see that $0001 \in \PPP$ and $0011 \in \SSS$. The suffix of size 3 of $0001$, namely $001$, is the same as the prefix of size $3$ of $0011$. Thus, when we append $1$ to $0001$ and we prepend $0$ to $0011$, we obtain the duplicate string $00011$.
\item  Finally, we eliminate one copy of each duplicate so as to balance the resulting sizes of $\PPP$ and $\SSS$ as much as possible. 
\end{enumerate}

The results in Table~\ref{tab} are obtained using the \textsf{Doubling 
Construction}. Note that here and elsewhere we denote the maximum size of a $(1,k)$-overlap-free code in $\{0,1\}^n$ by $C(n,k)$.

\begin{table}[t]
\caption{Results obtained from the \textsf{Doubling Construction}}
\label{tab}
\[\begin{array}{rrrr@{}l}
\hline
k   &    |P_k| &  |S_k|   &  
\multicolumn{2}{c}{C(k,n)\geq}\\
\hline
  2 & 2 & 1 & 2&\times 2^{n-4}\\
  3 & 3 & 2 & 6&\times 2^{n-6}\\
  4 & 5 & 4 & 20&\times 2^{n-8}\\
  5 & 8 & 8 & 64&\times 2^{n-10}\\
  6 & 15 & 14 & 210&\times 2^{n-12}\\
  7 & 26 & 27 & 702&\times 2^{n-14}\\
  8 & 50 & 50 & 2500&\times 2^{n-16}\\
  9 & 94 & 94 & 8836&\times 2^{n-18}\\
 10 & 180 & 179 & 32220&\times 2^{n-20}\\
 11 & 343 & 343 & 117649&\times 2^{n-22}\\
 12 & 659 & 659 & 434281&\times 2^{n-24}\\
 13 & 1267 & 1266 & 1604022&\times 2^{n-26}\\
 14 & 2444 & 2444 & 5973136&\times 2^{n-28}\\
 15 & 4726 & 4725 & 22330350&\times 2^{n-30}\\
 16 & 9157 & 9158 & 83859806&\times 2^{n-32}\\
 17 & 17779 & 17779 & 316092841&\times 2^{n-34}\\
 18 & 34575 & 34575 & 1195430625&\times 2^{n-36}\\
 19 & 67340 & 67339 & 4534608260&\times 2^{n-38}\\
 20 & 131323 & 131323 & 17245730329&\times 2^{n-40}\\
 21 & 256416 & 256416 & 65749165056&\times 2^{n-42}\\
 22 & 501208 & 501207 & 251208958056&\times 2^{n-44}\\
 23 & 980684 & 980684 & 961741107856&\times 2^{n-46}\\
\hline
\end{array}
\]
\end{table}

\section{Optimal solutions---a graph-based approach}
\label{graph.sec}

In this section, we discuss a graph-based approach that can (in principle) be used to prove that a solution is optimal. In practice, the method will only be feasible for small values of $k$. Again, we restrict our attention to the case $q=2$ for convenience. Denote $F = \{0,1\}$ and suppose $k$ is a fixed positive integer.

We construct a bipartite graph $G_k$. The vertex set is $X \cup Y$, where $|X| = |Y| = 2^k$.
We associate each vertex in $X$ with a $k$-tuple from $F^k$, and similarly each vertex in $Y$ corresponds to a $k$-tuple from $F^k$. The vertices in $X$ will be denoted by $x_p$, where $p \in F^k$, and the vertices in $Y$ will be denoted by $y_s$, where $s \in F^k$. We will join vertices $x_p$ and $y_s$ by an edge if and only if a prefix of $p$ is identical to a suffix of $s$.  For example, the graph $G_2$ is depicted in Figure \ref{fig1and2}.

In general, the graph $G_k$ records incompatible prefixes and suffixes. More precisely, if $x_py_s$ is an edge of $G_k$, then there cannot exist two $n$-tuples in a $(1,k)$-overlap-free code where $p$ is a $k$-prefix of an $n$-tuple and $s$ is a $k$-suffix of a (not necessarily distinct) $n$-tuple. 

The following lemma is immediate.

\begin{Lemma}
\label{graph.lem}
Suppose $C$ is a $(1,k)$-overlap-free code. Let $\PPP$ denote all the $k$-prefixes of $n$-tuples in $C$ and let $\SSS$ denote all the $k$-suffixes of $n$-tuples in $C$. Denote $X_C = \{ x_p : p \in \PPP\}$
and $Y_C = \{ y_s : s \in \SSS\}$. Then $X_C \cup Y_C$ is an independent set of vertices in $G_k$.
\end{Lemma}

\begin{Theorem}
\label{lowerg.thm}
Suppose $n \geq 2k$. Suppose that $X_C \cup Y_C$ is an independent set of vertices in $G_k$, where 
$X_C \subseteq X$ and $Y_C \subseteq Y$. Then there is a $(1,k)$-overlap-free code in $F^n$ having size 
\[
|X_C| \times |Y_C| \times 2^{n-2k}.
\]
\end{Theorem}

\begin{proof}
Suppose 
$X_C \cup Y_C$ is an independent set of vertices in $G_k$. Include 
all $n$-tuples of the form $p \parallel x \parallel s$ where $p \in \PPP$, $s \in \SSS$, and $x \in F^{n-2k}$. This  is a $(1,k)$-overlap-free code having size $|X_C| \times |Y_C| \times 2^{n-2k}$.
\end{proof}

\begin{Theorem}
\label{optimal.thm}
Suppose $n \geq 2k$. Suppose that $X_C \cup Y_C$ is an independent set of vertices in $G_k$, where 
$X_C \subseteq X$ and $Y_C \subseteq Y$, such that $|X_C| \times |Y_C|$ is maximized. Then the maximum size of any $(1,k)$-overlap-free code in $F^n$ is exactly $|X_C| \times |Y_C| \times 2^{n-2k}$.
\end{Theorem}

\begin{proof}
Suppose $C$ is a $(1,k)$-overlap-free code in $F^n$. Let $\PPP$ denote all the $k$-prefixes of $n$-tuples in $C$ and let $\SSS$ denote all the $k$-suffixes of $n$-tuples in $C$. Lemma \ref{graph.lem} asserts that 
$X_C \cup Y_C$ is an independent set of vertices in $G_k$. To maximize the size of $C$, we would include 
all $n$-tuples of the form $p \parallel x \parallel s$ where $p \in \PPP$, $s \in \SSS$, and $x \in F^{n-2k}$. From Theorem \ref{lowerg.thm}, this (optimal) code has size $|X_C| \times |Y_C| \times 2^{n-2k}$.
\end{proof}

\begin{Example}
{\rm
Suppose $k=2$. By examining the graph $G_2$ depicted in Figure \ref{fig1and2}, it is not hard to see that 
the only independent sets of size $4$ are $X$ and $Y$. Hence, the maximum value of $|X_C| \times |Y_C|$ is obtained when $|X_C|=2$ and  $|Y_C| = 1$ or when $|X_C|=1$ and  $|Y_C| = 2$. One optimal solution is
$X_C = \{x_{00},x_{01}\}$ and $Y_C = \{y_{11}\}$  (see the highlighted vertices in Figure \ref{fig1and2}). Therefore the maximum size of a $(1,2)$-overlap-free code in $F^n$ is $2^{n-3}$. In other words, the \textsf{Doubling Construction} is optimal for $k=2$.
}
\end{Example}

\begin{Remark}
{\rm
The proof of Theorem \ref{optimal.thm} uses the construction from Section \ref{sec3}. In Section \ref{sec3}, we inductively constructed independent sets $X_C \cup Y_C$ where we maximized 
$|X_C| \times |Y_C|$ at each step of the process. But it does not necessarily follow that the resulting values of $|X_C| \times |Y_C|$ are the maximum possible. In fact we will see situations where this is not the case.
}
\end{Remark}

\begin{figure}
\begin{center}

\begin{tikzpicture}[scale=0.35]

\draw [very thick] (0,0) -- (10,0) -- (0,5) -- (10,5) -- (0,10) -- (10,10)-- (0,0);
\draw [very thick] (10,5) -- (0,15) -- (10,15) -- (0,10);
\draw [very thick] (0,5) -- (10,10);
\draw [very thick,fill = white] (0,0) circle [radius=.25];
\draw [very thick,fill = white] (0,5) circle [radius=.25];
\draw [very thick,fill = black] (0,15) circle [radius=.25];
\draw [very thick,fill = black] (0,10) circle [radius=.25];
\draw [very thick,fill = black] (10,0) circle [radius=.25];
\draw [very thick,fill = white] (10,10) circle [radius=.25];
\draw [very thick,fill = white] (10,5) circle [radius=.25];
\draw [very thick,fill = white] (10,15) circle [radius=.25];

\node at (-2,0) {$x_{11}$};
\node at (-2,5) {$x_{10}$};
\node at (-2,10) {$x_{01}$};
\node at (-2,15) {$x_{00}$};

\node at (12,0) {$y_{11}$};
\node at (12,5) {$y_{10}$};
\node at (12,10) {$y_{01}$};
\node at (12,15) {$y_{00}$};

\end{tikzpicture}

\end{center}
\caption{The graph $G_2$, with nodes from an independent set highlighted}
\label{fig1and2}
\end{figure}
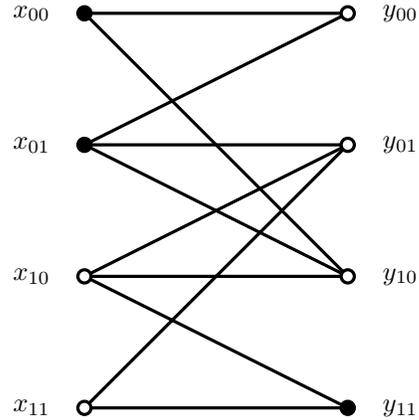

The graph $G_k$ has $2^{k+1}$ vertices. If we exhaustively search for an ``optimal'' independent set, this approach will quickly become infeasible as $k$ increases. This can be done for a few small values of $k$, however. The approach we take is to identify some nice structure in optimal independent sets for small $k$ and then generalize the structure to larger values of $k$. 

Suppose that $X_C \cup Y_C$ is an independent set of vertices in $G_k$, where $X_C \subseteq X$ and $Y_C \subseteq Y$. If $X_C \neq \emptyset$
and $Y_C \neq \emptyset$, then we say that $X_C \cup Y_C$ is a \emph{non-trivial} independent set. Now we present an  upper bound on the size of a non-trivial independent set in $G_k$. 

\begin{Theorem} 
\label{mis.thm} A non-trivial independent set in $G_k$ has size at most $2^{k-1}+1$.
\end{Theorem}

\begin{proof}
Define 
$X_i = \{x_p : p_1 = i\}$, for $i = 0,1$. Also, define 
$Y_i = \{y_s : s_k = i\}$, for $i = 0,1$. Thus $X_i$ consists of all vertices in $X$ corresponding to $k$-tuples beginning with $i$ and $Y_i$ consists of all vertices in $Y$ corresponding to $k$-tuples ending with $i$. Suppose that $X_C \cup Y_C$ is a non-trivial independent set of vertices in $G_k$; hence $X_C \neq \emptyset$
and $Y_C \neq \emptyset$. Suppose without loss of generality that there is an $x_p \in X_0 \cap X_C$. Then $Y_C \cap Y_0 = \emptyset$ and hence $Y_C \subseteq Y_1$. Since $Y_C \neq \emptyset$, we have $X_C \cap X_1 = \emptyset$ and hence $X_C \subseteq X_0$.

Therefore, we can restrict our attention to the subgraph $G'$ of $G$ induced by the vertices in 
$X_0 \cup Y_1$. $G'$ has $2^{k-1}$ vertices in each part of its partition.  We show that $G'$ contains a matching $M$ of size $2^{k-1}-1$. 

First, for the $2^{k-2}$ $k$-tuples $p$ such that $p_1 = 0$ and $p_k = 1$, we match $x_p$ with $y_p$.
The remaining $2^{k-2}$ $k$-tuples $p$ such that $x_p \in X_0$ have $p_1 = p_k = 0$ (call this set $\PPP'$), and the remaining $2^{k-2}$ $k$-tuples $s$ such that $y_s \in Y_1$ have $s_1 = s_k = 1$ (call this set $\SSS'$). We ignore the all-$0$ $k$-tuple in $\PPP'$ and the all-$1$ $k$-tuple in $\SSS'$; there remain $2^{k-2}-1$ $k$-tuples in $\PPP'$ and $2^{k-2}-1$ $k$-tuples in $\SSS'$.

Any $k$-tuple in $\PPP'$ can be written uniquely in the form $p = 0 \parallel \mathbf{a} \parallel 1 \parallel \mathbf{b} \parallel 0$, where
$\mathbf{a}$ is a (possibly empty) sequence of $0$'s and $\mathbf{b}$ is an arbitrary binary sequence. For each such $k$-tuple, we observe that there is an edge in $G'$ from $x_p$ to $y_s$, where $s = 1 \parallel \mathbf{b} \parallel 0 \parallel \mathbf{a} \parallel 1$, because $p$ begins with $0 \parallel \mathbf{a} \parallel 1$ and $s$ ends with $0 \parallel \mathbf{a} \parallel 1$. This creates $2^{k-2}-1$ additional matching edges.

We have constructed a matching of size $2^{k-1}-1$.
Since there are two unmatched vertices in $G'$, this immediately implies that the maximum size of a non-trivial independent set in $G'$ (and hence in $G_k$) is at most $2^{k-1}+1$.
\end{proof}

\begin{Remark}
The bound proven in Theorem \ref{mis.thm} is tight. This can be seen by observing that $\{ 00\cdots 0\} \cup Y_1$ is an independent set of size $2^{k-1}+1$.
\end{Remark}

\begin{Corollary}
\label{upper.cor}
For $k \geq 2$, it holds that 
\begin{align*}
C(k,n)&\leq  (2^{k-2}+1) \times 2^{k-2} \times 2^{n-2k}\\
&= 2^{n-4} + 2^{n-k-2}.
\end{align*}
\end{Corollary}

\begin{proof}
This is a straightforward application of Theorems \ref{optimal.thm} and \ref{mis.thm}. When $k \geq 2$, the value $2^{k-1}+1$ is odd. Therefore we maximize the product $|X_C| \times |Y_C|$ by taking 
\begin{align*}
|X_C| &= 2^{k-2}+1\\
\shortintertext{and}
|Y_C| &=  2^{k-2}
\end{align*}
(or vice versa).\qedhere
\end{proof}

We note that the upper bound proven in Corollary \ref{upper.cor} is weaker than the bound proven in Theorem \ref{1k.thm}.

\subsection{Results for small values of $k$}
\label{small.sec}

Let $I(k)$ denote the maximum size of a non-trivial independent set in $G_k$.
Table \ref{t0} summarizes the exact values of $I(k)$ and $C(k,n)$ for $k \leq 6$.

\begin{table}
\caption{Exact values of $I(k)$ and $C(k,n)$ for $k \leq 6$}
\label{t0}
\begin{center}
\begin{tabular}{c|c|r@{}l}
$k$ & $I(k)$ & 
\multicolumn{2}{c}{$C(k,n)$}  \\ \hline
$1$ & $2$ & &$2^{n-2}$ \\
$2$ & $3$ & $2$&$\times 2^{n-4}$ \\
$3$ & $5$ & $6$&$\times 2^{n-6}$  \\
$4$ & $9$ & $20$&$\times 2^{n-8}$  \\
$5$ & $16$ & $64$&$\times 2^{n-10}$ \\
$6$ & $30$ & $216$&$\times 2^{n-12}$  \\
\hline
\end{tabular}
\end{center}
\end{table}

It is clear that $I(1) = 2$ and the \textsf{Doubling Construction} is  optimal for $k=1$. 
Corollary \ref{upper.cor} shows that the \textsf{Doubling Construction} is optimal for $2 \leq k \leq 4$, and it also yields the exact values of $I(k)$ for these $k$. 

For $k=5$, an exhaustive search shows that $I(5) = 16$. From this, it follows that $C(5,n) \leq 8 \times 8 \times 2^{n-10} = 64 \times 2^{n-10}$. On the other hand, from the \textsf{Doubling Construction}, $C(5,n)\geq 2^{n-4}=64\times 2^{n-10}$, and so the \textsf{Doubling Construction} is again optimal. 
For $k=6$, Theorem \ref{mis.thm} shows that $I(6) \leq 33$ and Corollary \ref{upper.cor} states that 
\[C(n,6) \leq 2^{n-4} + 2^{n-8} = 272 \times 2^{n-12}.\]
However, this is not a tight bound, as we discuss below. The \textsf{Doubling Construction} yields a non-trivial independent set of size $29$ with $14$ vertices in one part and $15$ vertices in the other part. Hence, 
\[C(6,n) \geq 15 \times 14 \times 2^{n-12} 
= 210 \times 2^{n-12}.\]
But it turns out that there is a non-trivial independent set of size $30$ with $12$ vertices in one part and $18$ vertices in the other part. This leads to a larger $(1,6)$-overlap-free code because
$18 \times 12 > 15 \times 14$. The resulting lower bound is 
\[
C(6,n) \geq 18 \times 12 \times 2^{n-12} = 216 \times 2^{n-12}.
\]
This solution is in fact optimal, as was verified by an exhaustive search. 
Here are the $6$-tuples in the sets $\PPP$ and $\SSS$:

\[
\begin{array}{c}
\PPP \\ \hline
000000, 000001, 000010, 000011,\\
000100, 000101, 000110, 000111,\\
001000, 001001, 001010, 001011 
\end{array} 
\]

\[
\begin{array}{c}
\SSS \\ \hline
001101, 001111, 010011, 010101,\\
010111, 011011, 011101, 011111,\\
100111, 101011, 101101, 101111,\\
110011, 110101, 110111, 111011,\\
111101, 111111
\end{array} 
\]

\section{The \textsf{$m$-minimum Construction}}
\label{m-min.sec}
\begin{table}
\caption{Results obtained from the \textsf{$m$-minimum Construction}}
\label{tab2}
\[\begin{array}{rrrr@{}l}
\hline
k   &    |\PPP| &  |\SSS|   &  
\multicolumn{2}{c}{C(k,n) \geq}\\
\hline
  2 & 1 & 2 & 2&\times 2^{n-4}\\
  3 & 2 & 3 & 6&\times 2^{n-6}\\
  4 & 4 & 5 & 20&\times 2^{n-8}\\
  5 & 8 & 8 & 64&\times 2^{n-10}\\
  6 & 12 & 18 & 216&\times 2^{n-12}\\
  7 & 24 & 31 & 744&\times 2^{n-14}\\
  8 & 44 & 60 & 2640&\times 2^{n-16}\\
  9 & 64 & 149 & 9536&\times 2^{n-18}\\
 10 & 128 & 274 & 35072&\times 2^{n-20}\\
 11 & 256 & 504 & 129024&\times 2^{n-22}\\
 12 & 512 & 927 & 474624&\times 2^{n-24}\\
 13 & 960 & 1823 & 1750080&\times 2^{n-26}\\
 14 & 1792 & 3644 & 6530048&\times 2^{n-28}\\
\hline
\end{array}
\]
\end{table}
For $k \geq 7$, exhaustive searches appear to be infeasible. So we have tried various 
techniques to find useful lower bounds. 
We first describe the \textsf{$m$-minimum Construction}, which has enabled us to find some good solutions. 
\begin{Construction}[\textsf{$m$-minimum Construction}]
\label{m-min.const}
Suppose $k$ is a given positive integer. For $m = 1, 2, \dots, 2^{k-1}$, we construct a code $D_m$ as follows:
\begin{itemize}
\item Let $\PPP$ consist of the first $m$ non-negative integers, represented as binary $k$-tuples (padded on the left with $0$'s if necessary, i.e., in big-endian form). Define 
$X_C = \{ x_p : p \in \PPP\}$.
\item Let $Y_C$ consist of all vertices in $Y$ that are adjacent to no vertices in $X_C$. 
Define $\SSS = \{s :  y_s \in Y_C\}$. 
\item Output the sets $\PPP$ and $\SSS$ for the code $D_m$ that maximizes the value of $|\PPP| \times |\SSS|$. The resulting $(1,k)$-overlap-free code will have size $|\PPP| \times |\SSS| 
\times 2^{n-2k}$.
\end{itemize}
\end{Construction}
\vbox{
Table \ref{tab2} summarizes results obtained from the \textsf{$m$-minimum Construction}. For $k \geq 6$, these are all improvements over the \textsf{Doubling Construction}.
The optimal solution for $k=6$ that we presented in Section \ref{small.sec} is precisely the code $D_{12}$ obtained from the \textsf{$m$-minimum Construction}.
For $k=7$, $D_{24}$ is the code found by the \textsf{$m$-minimum Construction}; it has $|\PPP| = 24$ and $|\SSS| = 31$:
}
\[
\begin{array}{c}
\PPP \\ \hline
0000000,0000001,0000010,0000011,\\
0000100,0000101,0000110,0000111,\\
0001000,0001001,0001010,0001011,\\
0001100,0001101,0001110,0001111,\\
0010000,0010001,0010010,0010011,\\
0010100,0010101,0010110,0010111
\end{array} 
\]
and
\[
\begin{array}{c}
\SSS \\ \hline
0011011,0011101,0011111,0100111,\\
0101011,0101101,0101111,0110011,\\
0110101,0110111,0111011,0111101,\\
0111111,1001101,1001111,1010011,\\
1010101,1010111,1011011,1011101,\\
1011111,1100111,1101011,1101101,\\
1101111,1110011,1110101,1110111,\\
1111011,1111101,1111111
\end{array} 
\]
This yields the lower bound \[C(7,n) \geq 744 \times 2^{n-14}.\]

\subsection{The \textsf{Zero Block Construction}}
\label{simon.const}
We now present the \textsf{Zero Block Construction}, which is closely related to the \textsf{$m$-minimum Construction}, and is inspired by the classical construction of non-overlapping codes due to Gilbert and Levenshtein~\cite{Gil60,Lev64,Lev70} which we discuss in Section~\ref{classical.sec}.
\begin{Construction}[\textsf{Zero Block Construction}]
Suppose $k$ is a given positive integer. For $z=1 , \dots ,  k-1$, we construct a code $C_z$ from a certain $X_C$ and $Y_C$ as follows:
\begin{itemize}
\item Let $\PPP$ consist of the first $2^{k-z}$ non-negative integers, represented as binary $k$-tuples. Note that every $p \in \PPP$ begins with a block of (at least) $z$ consecutive $0$'s. Define $X_C = \{ x_p : p \in \PPP\}$.
\item Let $\SSS$  consist of all binary $k$-tuples $s$ ending with a $1$ that do not contain $z$ consecutive $0$'s. 
Define 
$Y_C = \{ y_s : s \in \SSS\}$.
\item Output the sets $\PPP$ and $\SSS$ for the code $C_z$ that maximizes the value of $|\PPP| \times |\SSS|$. The resulting $(1,k)$-overlap-free code will have size 
\[ |\PPP| \times |\SSS| \times 2^{n-2k} 
=  |\SSS| \times 2^{n-k-z}.
\]
\end{itemize}
\end{Construction}

\begin{Lemma}
For $X_C$ and $Y_C$ as defined in Construction \ref{simon.const}, no vertex in $Y_C$ is adjacent to any vertex in $X_C$.
\end{Lemma}

\begin{proof}
Suppose $x_p \in X_C$ and $y_s \in Y_C$. We consider two cases. If $\ell \leq z$, then the $\ell$-prefix of $p$ consists of $\ell$ $0$'s. However, $s$ ends in a $1$, so the $\ell$-suffix of $s$  is not the same as the $\ell$-prefix of $p$. The second case is when $\ell \geq z+1$. Here an $\ell$-prefix of $p$ begins with  $z$ $0$'s. However, no $\ell$-suffix of $s$ contains $z$ consecutive $0$'s, so the $\ell$-suffix of $s$  is not the same as the $\ell$-prefix of $p$.
\end{proof}

 Thus, for any fixed value of $z$, the set $Y_C$ defined in Construction \ref{simon.const} is a subset of the set that would be chosen in Construction \ref{m-min.const} (the \textsf{$m$-minimum Construction}). So the \textsf{Zero Block Construction} cannot improve on the \textsf{$m$-minimum Construction}; however, it is an explicit construction and potentially easier to analyze.
We will consider a general bound that can be proven, as well as numerical computations for various values of $k$.

It remains to specify an appropriate value for $z$ and to investigate the size of $\SSS$.
It turns out that the number of 
binary $\ell$-tuples $s$ that do not contain $n$ consecutive $0$'s is given by an \emph{$n$-step Fibonacci number}. For a given value of $n \geq 2$, the \emph{$n$-step Fibonacci sequence} 
is defined recursively as follows. 
\begin{equation}
\label{fib.eq} F_{i}^{(n)}=
\begin{cases}
0 & \text{if $-n+2 \leq i \leq 0$}\\
1 & \text{if $i = 1$}\\
\displaystyle \sum_{j=1}^n F_{i-j}^{(n)} & \text{if $i \geq 2$.}
\end{cases}
\end{equation}
That is, each term in this sequence is the sum of the $n$ previous terms. 
It is easy to see that
\[F_{i}^{(n)} = 2^{i-2}\]
for $2 \leq i \leq n+1$. Also, it is easily verified that  
\[F_{n+2}^{(n)} = 2^{n}-1 \quad  \text{ and } \quad  F_{n+3}^{(n)} = 2^{n+1}-3.\]
For additional information about these sequences, see \cite{Dres,Wolf}.

The following  result is well-known. We provide a proof for completeness.
\begin{Lemma}
\label{seq.lem}
The number of 
binary $\ell$-tuples that do not contain $z$ consecutive $0$'s is $F_{\ell+2}^{(z)}$.
\end{Lemma}

\begin{proof}
Denote the number of 
binary $\ell$-tuples that do not contain $z$ consecutive $0$'s by $g(\ell,z)$. Then it is clear that
$g(\ell,z) = 2^{\ell}$, if $1 \leq \ell < z$,
and
$g(z,z)=2^z - 1.$
Thus $g(\ell,z) = F_{\ell+2}^{(z)}$ if
$1 \leq \ell \leq z$.

Next, consider $g(\ell,z)$ for some $\ell > z$. We partition the set of all binary $\ell$-tuples that do not contain $z$ consecutive $0$'s into $z$ disjoint subsets, denoted by $W_i$, $i = 1 , \dots , z$. For $1 \leq i \leq z$, the set $W_i$ consists of all the $\ell$-tuples that end with a $1$ followed by $i-1$ $0$'s. It is clear that $|W_i| = g(\ell-i,k)$ for $1 \leq i \leq z$. Hence, 
\[ g(\ell,z) = \sum_{i=1}^z g(\ell-i,z)\] whenever $\ell > z$.
We can assume by induction that $g(\ell-i,z) = F_{\ell-i+2}^{(z)}$ for $1 \leq i \leq z$.
So \[ g(\ell,z) = \sum_{i=1}^z F_{\ell-i+2}^{(z)} = F_{\ell+2}^{(z)},\]
from (\ref{fib.eq}), as desired.
\end{proof}

The number of choices for $s \in \SSS$ is exactly $F_{k+1}^{(z)}$. Thus we have the following result.

\begin{Theorem}
\label{zeroblock.thm}
The size of the code obtained from the \textsf{Zero Block Construction} is 
\begin{equation}
\label{zeroblock.eq}
 \max \left\{ F_{k+1}^{(z)} \times 2^{n-k-z} : 1 {\leq} z {\leq} k-1 \right\} .
\end{equation}
\end{Theorem}

In order to obtain an explicit closed-form bound, it is probably more convenient to  work with a simple lower bound on 
the values $F_{k+1}^{(z)}$. 

\begin{Lemma}
\label{events.lem}
For $1\leq z \leq k-1$, the following bound holds:
 \[ F_{k+1}^{(z)} >(1-k\, 2^{-z}) 2^{k-1}.\] 
\end{Lemma}
\begin{proof}
Choose a binary word $t$ of length $k-1$ randomly and uniformly and then append a $1$. Let $E_i$ be the `bad' event that $t$ contains $0^z$, starting at position~$i$. Note that $t$ is of the desired form if and only if none of the events $E_1,E_2,\ldots,E_{k-1}$ occur. But the probability of $E_i$ is at most $2^{-z}$ (indeed it is equal to this when $i\leq k-z$, and it is $0$ otherwise). So the probability that one or more of the $E_i$'s occurs is at most 
$(k-1)2^{-z}$. Hence the probability that none of the events $E_1,E_2,\ldots,E_{k-1}$ occur is at least $1-(k-1)/2^{z}$.
Since \[1-(k-1)2^{-z} > 1-k\, 2^{-z},\] the stated bound follows.
\end{proof}

Now, using equation (\ref{zeroblock.eq}) from Theorem \ref{zeroblock.thm},  
for a given value of $z$, we obtain a code of size at least
\begin{align*}
 (1-k\,&2^{-z})\times 2^{k-1} \times 2^{n-k-z} \\
&=  (1-k\, 2^{-z})\times 2^{n-z-1}\\
&= (2^{-z}(1-k\, 2^{-z}))2^{n-1}.
\end{align*}
The  function $f(z) = 2^{-z}(1-k\, 2^{-z})$ is maximized when $z=\log_2 2k$. Sadly, this is not always an integer. However, taking $z_0= \displaystyle \lfloor \log_2 2k \rceil $ (i.e., rounding $\log_2 2k$ to the nearest integer), we have 
\[\log_2 2k - 1/2 \leq z_0 \leq \log_2 2k + 1/2,\] so 
$2^{z_0}\in [ \sqrt{2}k,2\sqrt{2}k]$.  
It then follows that 
\[
f(z_0) \geq
\max \{ f(\log_2 2k - \tfrac{1}{2}) , f(\log_2 2k + \tfrac{1}{2})  \} .
\]
We have 
\begin{align*}
f(\log_2 2k - \tfrac{1}{2}) 
&= \frac{1}{\sqrt{2}k}\left(  1 - \tfrac{1}{\sqrt{2}} \right)
\approx \frac{1}{4.83k}\\
\intertext{and}
f(\log_2 2k + \tfrac{1}{2})
&= \frac{1}{2\sqrt{2}k}\left(  1 - \tfrac{1}{2\sqrt{2}} \right)
\approx \frac{1}{4.38k}.
\end{align*}

Hence, $f(z_0) \geq 1/(4.83k)$. Since the size of the resulting code is $f(z_0) \times 2^{n-1}$, we have the following theorem.

\begin{Theorem}
\label{gen1.thm}
There exists $z$ such that \[|C_z|>(1/9.67k)2^n;\] hence \[C(k,n) \geq (1/9.67k)2^n.\]
\end{Theorem}

We now incorporate two tweaks to improve Theorem \ref{gen1.thm}.
The first is to define the events $E_1, E_2, \dots$ used in the proof of Lemma \ref{events.lem} a bit more carefully.

\begin{Lemma}
\label{events2.lem}
For $1\leq z \leq k-1$, the following bound holds:
 \[ F_{k+1}^{(z)} \geq (1-k\, 2^{-z-1}) 2^{k-1}.\] 
\end{Lemma}
\begin{proof}
As before, choose a binary word $t$ of length $k-1$ randomly and uniformly and then append a $1$. We define $E_1$ as before. However, for $2 \leq i \leq k-z$, we now define $E_i$ to be the event that there is a $1$ in position $i-1$, followed by $z$ $0$'s. It is not hard to see that if $t$ contains $z$ consecutive zeroes, then one of the events $E_1, \dots , E_{k-z}$ occurs. This is because the first occurrence of $z$ consecutive $0$'s must immediately follow a $1$, except when the first $z$ positions are all $0$'s.

We have $\Pro [E_1] = 2^{-z}$ and  $\Pro [E_i] = 2^{-z-1}$ for $2 \leq i \leq k-z$. Hence,
\begin{align*} 
\Pro [E_1 \vee \cdots \vee E_{k-z}] &\leq 2^{-z} 
+ (k{-}z{-}1)2^{-z-1}\\
 &= 2^{-z-1} (2 + k-z-1)\\
 &\leq  k\, 2^{-z-1},
\end{align*}
since $z \geq 1$.
Hence, 
\begin{align*} \Pro [\overline{E_1} \wedge \cdots \wedge \overline{E_{k-z}}] &\geq 1 - k\, 2^{-z-1}.
\end{align*}
The stated bound follows.
\end{proof}
Using equation (\ref{zeroblock.eq}) from Theorem \ref{zeroblock.thm},  
for a given value of $z$, we obtain a code of size at least 
\[
(2^{-z}(1-k\, 2^{-z-1}))2^{n-1}.\] 
In order to maximize the size of the code, we choose $z$ to maximize the function
\[g(z) = 2^{-z}(1-k\, 2^{-z-1}). \]
The maximum occurs when  $z=\log_2 k$, which of course might not be an integer.
We could consider an interval of length $1$ whose centre is at $\log_2 k$ (similar to our argument above), but we can do slightly better by considering a different interval (this is our second tweak). 

We choose $z$ to be an integer in the interval $\left[ \log_2 \frac{3k}{4}, \log_2 \frac{3k}{2} \right].$ Notice that this is again an interval of length $1$. We obtain a slightly better bound because $g\left(\log_2 \frac{3k}{4}\right) = g\left(\log_2 \frac{3k}{2}\right)$. In fact, 
\[ g\left(\log_2 \frac{3k}{4}\right) = g\left(\log_2 \frac{3k}{2}\right) = \frac{4}{9k} .\]
We immediately obtain the following theorem, which improves Theorem \ref{gen1.thm}.

\begin{Theorem}
\label{gen2.thm}
There exists $z$ such that 
\[ |C_z|>(2/9k)2^n;\]
 hence $C(k,n) \geq (2/9k)2^n$.
\end{Theorem}

When $k$ is a power of $2$, the function $g(z)$ is maximized at the integral value $z=\log_2 k$.
We obtain an improved result in this case.
\begin{Theorem}
\label{gen3.thm}
If  $k = 2^i$ for a positive integer $i$, then 
\[|C_i|> (1/4k)2^n;\] 
hence $C(k,n) {\geq} (1/4k)2^n$ for these values of $k$.
\end{Theorem}

We note that the upper bound from Theorem \ref{1k.thm} is   $C(k,n) \leq (1/2k)2^n$, which is roughly a factor of two 
greater than the lower bound from Theorem \ref{gen3.thm} (when $k$ is a power of two).

It is also possible to obtain asymptotic bounds which are stronger than the explicit general bounds discussed above. We pursue this now.

Let $\ell$ and $z$ be integers, with $1\leq z<\ell$. For an integer $k$ with $0\leq k<\ell$, define $\phi(\ell,k,z)$ to be the number of binary sequences of length $\ell$ and weight~$k$ such that any two cyclically consecutive ones are separated by at least $z$ zeros. The following lemma gives bounds for $\phi(\ell,k,z)$ that are good when $k$ and $z$ are small compared to $\ell$:

\begin{Lemma}
\label{lem:phi_estimate}
Define $\ell$, $z$, $k$ and $\phi(\ell,k,z)$ as above. Then
\[
\binom{\ell}{k}-kz\binom{\ell}{k-1}\leq \phi(\ell,k,z)\leq \binom{\ell}{k}.
\]
\end{Lemma}
\begin{proof}
The lemma follows trivially in the case when $k\leq 1$, since $\phi(\ell,0,z)=1$ and $\phi(\ell,1,z)=\ell$. So we may assume that $k\geq 2$.

The upper bound follows since $\binom{\ell}{k}$ is the number of weight $k$ binary sequences of length $\ell$. The lower bound follows if we can show that there are at most $kz\binom{\ell}{k-1}$ weight $k$ binary sequences of length $\ell$ that have a zero run of length less than $z$. But all such sequences can be obtained (possibly more than once) in the following three-stage process. In Stage~1, choose a set of $k-1$ positions in the sequence to be equal to $1$. In Stage~2, choose one of these  $k-1$ positions, say position $i$. In Stage~3, choose a position $i+a\bmod\ell$ where $1\leq a\leq z$ and set this position equal to $1$; set the remaining positions to be zero. There are at most $\binom{\ell}{k-1}$ choices in the first stage, there are $k-1$ choices in the second stage and at most $z$ choices in the third stage. So
\begin{align*}
\phi(\ell,k,z)&\geq \tbinom{\ell}{k}-\tbinom{\ell}{k-1}(k-1)z\\[1em]
&\geq \tbinom{\ell}{k}-kz\tbinom{\ell}{k-1},
\end{align*}
as required.
\end{proof}

\begin{Corollary}
\label{cor:phi_binom}
Define $\ell$, $z$, $k$ and $\phi(\ell,k,z)$ as above. Then
\[
\frac{1}{k!}-\frac{2kz}{\ell}\leq
\frac{\phi(\ell,k,z)}{\ell^k}\leq \frac{1}{k!}
\]
\end{Corollary}
\begin{proof}
The upper bound follows from the upper bound of Lemma~\ref{lem:phi_estimate} and the 
inequality \[\tbinom{\ell}{k}\leq \ell^k/k!.\] For the lower bound, we use the lower bound of Lemma~\ref{lem:phi_estimate} and the same bound on a binomial coefficient to see that
\begin{align*}
\frac{\phi(\ell,k,z)}{\ell^k}&\geq \frac{\tbinom{\ell}{k}}{\ell^k} -\frac{zk}{\ell} \times \frac{\tbinom{\ell}{k-1}}{\ell^{k-1}}\\
&\geq \frac{\tbinom{\ell}{k}}{\ell^k}-\frac{kz}{(k-1)!\ell}\\
&\geq \frac{\tbinom{\ell}{k}}{\ell^k}-\frac{kz}{\ell}.
\end{align*}
The lower bound now follows since
\begin{align*}
\tbinom{\ell}{k}&\geq \tfrac{(\ell-k)^k}{k!}
\geq \tfrac{\ell^k-k^2\ell^{k-1}}{k!}\\
&\geq \tfrac{\ell^k}{k!} - k\ell^{k-1}
\geq \tfrac{\ell^k}{k!} -kz\ell^{k-1}. \qedhere
\end{align*} 
\end{proof}

\begin{Theorem}
\label{thm:cyclic_zeros}
For a positive integer $a$, define $\ell=2^a$ and $z=a-1$ \textup{(}so $2^{z+1}=\ell$\textup{)}. Let $\nu_a$ be the number of binary sequences of length $\ell$ that do not contain any cyclic runs of $z$ or more consecutive zeros. Then $\lim_{a\rightarrow\infty}\nu_a/2^\ell=1/e$ \textup{(}where $e$ is the base of the natural logarithm\textup{)}.
\end{Theorem}
\begin{proof}
Let $X_i$ be the set of sequences 
\[
s = (s_0,s_1,\ldots,s_\ell)
\]
such that $s_i{=}1$ and $s_{i+1}{=}s_{i+2}{=}\cdots{=}s_{i+z}{=}0$. (Here we take subscripts modulo $\ell$.)

Note that $s$ has no cyclic runs of $z$ or 
more zeros if and only if $s$ is non-zero 
and $s\not\in X_i$ for $i \in L= 
\{ 0,1,2,...,\ell-1\}$. Hence
\begin{equation}
\label{eqn:nu}
\nu_a=\biggl|\,\overline{\bigcup_{i\in L} X_i}\,\biggr|-1.
\end{equation}

By the principle of inclusion-exclusion,
\begin{equation}
\label{eqn:PIE}
\biggl|\,\overline{\bigcup_{i \in L} X_i}\,\biggr|
=\sum_{k=0}^{\ell-1}\,(-1)^k
\sum_{\substack{I\subseteq L\\|I|=k}}\biggl|\bigcap_{i\in I} X_i\;\biggr|,
\end{equation}
where the partial sums involving $k$ on the right hand side are successively upper and lower bounds for the left hand side (this follows from the Bonferroni inequalities).

For a subset $I\subseteq L$, let 
\[t^I=(t^I_0,t^I_1,\ldots,t^I_{\ell-1})\]
be the indicator binary sequence for $I$, so
\[
t^I_i=\begin{cases} 1&\text{ if }i\in I,\\
0&\text{ otherwise}.
\end{cases}
\]
When $k\geq 1$ we see that
\[
\biggl|\bigcap_{i\in I} X_i\biggr|
= 2^{\ell-(z+1)k},
\]
when any two cyclically consecutive
ones in $t^I$
are separated
by at least $z$ zeros, and is $0$
otherwise.
So using the notation above Lemma~\ref{lem:phi_estimate}, we may simplify~\eqref{eqn:PIE} as
\vbox{
\begin{align}
\nonumber
\biggl|\,\overline{\bigcup_{i \in L} X_i}\,\biggr|
&=\displaystyle \sum_{k=0}^{\ell-1}(-1)^k \phi(\ell,k,z)2^{\ell-(z+1)k}\\
&=\displaystyle 2^\ell\left(\sum_{k=0}^{\ell-1}(-1)^k \tfrac{\phi(\ell,k,z)}{\ell^k}\right),
\label{eqn:simple_PIE}
\end{align}
since $2^{z+1}=\ell$ by our choice of $\ell$ and $z$.
}

Define $b=\lfloor \ell^{1/4}\rfloor$. We noted above that successive partial sums in the right hand side of~\eqref{eqn:simple_PIE} are upper and lower bounds for the left hand side, so truncating this sum after $b+1$ terms, we see that 
\begin{multline} 
\biggl|\,\overline{\bigcup_{i \in L} X_i}\,\biggr|
\left.-2^\ell\left(\sum_{k=0}^{b}(-1)^k \frac{\phi(\ell,k,z)}{\ell^k}\right)\right|\\[1em]
\leq 2^\ell \frac{\phi(\ell,b+1,z)}{\ell^{b+1}}
\leq \frac{2^\ell}{(b+1)!},
\label{eqn:PIE_truncate}
\end{multline}
the final inequality following by the upper bound of Corollary~\ref{cor:phi_binom}. Now, $z<b$ when $a$ is sufficiently large, since $z=a-1$ and $b\geq 2^{a/4}-1$. 
So, using the bounds in Corollary~\ref{cor:phi_binom}, 
\begin{align}
\nonumber
\left| 2^\ell\left(\sum_{k=0}^{b}(-1)^k\right.\right.
&\left.\left.\!\!\tfrac{\phi(\ell,k,z)}{\ell^k}\right)-2^\ell\sum_{k=0}^b \tfrac{(-1)^k}{k!}\right|\\
\nonumber
&\leq 2^\ell\sum_{k=0}^b\tfrac{2kz}{\ell}
\leq 2^\ell \sum_{k=1}^b\tfrac{2b^2}{\ell}\\
\label{eqn:two_truncations}
&=2^\ell\tfrac{2b^3}{\ell}\\&< 2^\ell\tfrac{2}{\ell^{1/4}}
\end{align}
whenever $a$ is sufficiently large. But the usual power series expansion for $1/e$ shows that
\begin{equation}
\label{eqn:e_approximation}
2^\ell\left|\frac{1}{e}-\sum_{k=0}^b\frac{(-1)^k}{k!}\right|\leq \frac{2^\ell}{(b+1)!}.
\end{equation}
Combining equations~\eqref{eqn:nu}, \eqref{eqn:PIE_truncate}, \eqref{eqn:two_truncations} and~\eqref{eqn:e_approximation} 
we see that $\nu_a=2^\ell(1/e+\epsilon)$,
where
\[
|\epsilon|\leq \frac{1}{2^\ell}+\frac{2}{(b+1)!}+\frac{2}{\ell^{1/4}}
\]
whenever $a$ is sufficiently large. In particular $\epsilon$ tends to zero as $a\rightarrow\infty$, and so the theorem follows.
\end{proof}
We remark that Schoeny \emph{et al.}~\cite[Subsection~V.B]{Schoe} prove a bound on the number of binary sequences with no (zero or one) runs of length $\log(2n)$, using a probabilistic construction. We wonder whether their bounds could be improved using techniques similar to those in the proof of Theorem~\ref{thm:cyclic_zeros}.

\begin{Corollary}
\label{cor:noncyclic_zeros}
For a positive integer $a$, define $\ell=2^a$ and $z=a-1$. Then 
\[\lim _{a \rightarrow \infty}  \frac{F_{\ell+2}^{(z)} }{2^{\ell}}   =   \frac{1}{e}\]
(where $e$ is the base of the natural logarithm).
\end{Corollary}
\begin{proof}
Recall that $\nu_a$ is the number of binary sequences of length $\ell$ that do not contain any cyclic runs of $z$ or more consecutive zeros. Also, $F_{\ell+2}^{(z)}$ is the number of binary sequences of length $\ell$ that do not contain any  runs of $z$ or more consecutive zeros. Clearly \[\nu_a\leq F_{\ell+2}^{(z)}.\] A sequence with no (non-cyclic) runs of $z$ or more consecutive zeros, but which contains a cyclic run of $z$ or more zeros, must either start or end with at least $\lceil a/2\rceil$ zeros. Hence $F_{\ell+2}^{(z)} -\nu_a\leq 2(2^{\ell-\lceil a/2\rceil})$. Hence
\[
2^{-\ell}\left|\nu_a-F_{\ell+2}^{(z)} \right|\leq 2^{1-\lceil a/2\rceil}.
\]
Since the right hand side of this inequality tends to $0$ as $a\rightarrow\infty$, the corollary follows by Theorem~\ref{thm:cyclic_zeros}.
\end{proof}

We can now prove the following asymptotic lower bound on $C(k,n)$.
 \begin{Theorem}
 \[\limsup_{k \rightarrow \infty} \frac{C(k,n)}{2^n} \geq \frac{1}{ek} .\]
 \end{Theorem}
 \begin{proof}
 Take $k = \ell+1 = 2^a+1$ in equation (\ref{zeroblock.eq}) from Theorem \ref{zeroblock.thm}. 
Since $z = a-1$, we now have $2^{z+1} = k-1 = \ell$. Then Theorem \ref{zeroblock.thm} says that
\begin{align*}
\limsup _{k \rightarrow \infty}& \frac{C(k,n)}{2^n} \\
&\geq
 \lim _{k \rightarrow \infty} \frac{F_{k+1}^{(z)}}{2^{k+z}}\\
&= \lim _{\ell \rightarrow \infty} \frac{F_{\ell+2}^{(z)}}{2^{\ell+1+z}}\\
&= 
\frac{1}{e\,2^{z+1}} 
\quad\text{from Corollary \ref{cor:noncyclic_zeros}}\\
&  = \frac{1}{e(k-1)}\\
 &> \frac{1}{ek}.\qedhere
 \end{align*}
 \end{proof}

\begin{table*}[t]
\caption{Comparison of lower bounds obtained from the \textsf{$m$-minimum Construction} and the \textsf{Zero Block Construction}}
\label{tab3}
\[\begin{array}{r@{\hspace{2em}}r@{}l@{\hspace{6em}}r@{}l@{\hspace{6em}}c}
\hline
k
& \multicolumn{2}{@{}c@{}}{\textsf{$m$-minimum Construction}}   
& \multicolumn{2}{@{}c@{}}{\textsf{Zero Block Construction}}
& \text{optimal value of $z$}\\
\hline
  2 &  2&\times 2^{n-4} & 2&\times 2^{n-4} & 1\\
  3 &  6&\times 2^{n-6} & 6&\times 2^{n-6} & 2\\
  4 &  20&\times 2^{n-8} & 20&\times 2^{n-8} & 2\\
  5 &   64&\times 2^{n-10} & 64&\times 2^{n-10} & 2\\
  6 &   216&\times 2^{n-12} & 208&\times 2^{n-12} & 2\\
  7 &   744&\times 2^{n-14} & 704&\times 2^{n-14} & 3\\
  8 &   2640&\times 2^{n-16}\ & 2592&\times 2^{n-16} & 3\\
  9 &   9536&\times 2^{n-18} & 9536&\times 2^{n-18} & 3\\
 10 &   35072&\times 2^{n-20} & 35072&\times 2^{n-20} & 3\\
 11 &   129024&\times 2^{n-22} & 129024&\times 2^{n-22} & 3\\
 12 &   474624&\times 2^{n-24} & 474624&\times 2^{n-24} & 3\\
 13 &   1750080&\times 2^{n-26}& 1745920&\times 2^{n-26} & 3\\
 14 &  6530048&\times 2^{n-28} & 6422528&\times 2^{n-28} & 3\\
\hline
\end{array}
\]
\end{table*}

\begin{table*}[t]
\begin{center}
\[
\begin{array}{
r
r@{}l
r@{}l
r@{}l
|
r@{}l
}
\hline
k&
\multicolumn{2}{@{}c@{}}{\text{Doubling}}&
\multicolumn{2}{@{}c@{}}{\text{$m$-minimum}}&
\multicolumn{2}{@{}c@{}}{\text{Zero block}}&
\multicolumn{2}{@{}c@{}}{\text{Upper bound}}\\
\hline
2 &  \mathbf{2}&\mathbf{\times 2^{n-4}}&  \mathbf{2}&\mathbf{\times 2^{n-4}} &  \mathbf{2}&\mathbf{\times 2^{n-4}}&\mathbf{2}&\mathbf{\times 2^{n-4}}\\
3 &   \mathbf{6}&\mathbf{\times 2^{n-6}}& \mathbf{6}&\mathbf{\times 2^{n-6}} &  \mathbf{6}&\mathbf{\times 2^{n-6}}&\mathbf{6}&\mathbf{\times 2^{n-6}}\\
4 &   \mathbf{20}&\mathbf{\times 2^{n-8}} & \mathbf{20}&\mathbf{\times 2^{n-8}} &  \mathbf{20}&\mathbf{\times 2^{n-8}}&\mathbf{20}&\mathbf{\times 2^{n-8}} \\
5 &  \mathbf{64}&\mathbf{\times 2^{n-10}}&   \mathbf{64}&\mathbf{\times 2^{n-10}} &  \mathbf{64}&\mathbf{\times 2^{n-10}}&\mathbf{64}&\mathbf{\times 2^{n-10}}\\
6 & 210&\times 2^{n-12}&  \mathbf{ 216}&\mathbf{\times 2^{n-12}} & 208&\times 2^{n-12}&\mathbf{216}&\mathbf{\times 2^{n-12}}\\
7 &  702&\times 2^{n-14} &  \mathbf{744}&\mathbf{\times 2^{n-14}} & 704&\times 2^{n-14} &1170.3&\times 2^{n-14}\\
8 &  2500&\times 2^{n-16}&  \mathbf{2640}&\mathbf{\times 2^{n-16}} & 2592&\times 2^{n-16} &4096&\times 2^{n-16}\\
9 & 8836&\times 2^{n-18}&   \mathbf{9536}&\mathbf{\times 2^{n-18}} &  \mathbf{9536}&\mathbf{\times 2^{n-18}} &14563.6&\times 2^{n-18}\\
10 & 32220&\times 2^{n-20}&   \mathbf{35072}&\mathbf{\times 2^{n-20}} &  \mathbf{35072}&\mathbf{\times 2^{n-20}} &52428.8&\times 2^{n-20}\\
11 & 117649&\times 2^{n-22}&   \mathbf{129024}&\mathbf{\times 2^{n-22}} &  \mathbf{129024}&\mathbf{\times 2^{n-22}} &190650.2&\times 2^{n-22}\\
12 & 434281&\times 2^{n-24}&   \mathbf{474624}&\mathbf{\times 2^{n-24}} &  \mathbf{474624}&\mathbf{\times 2^{n-24}} &699050.7&\times 2^{n-24}\\
13 & 1604022&\times 2^{n-26}&   \mathbf{1750080}&\mathbf{\times 2^{n-26}}& 1745920&\times 2^{n-26} &2581110.2&\times 2^{n-26}\\
14 & 5973136&\times 2^{n-28}&  \mathbf{6530048}&\mathbf{\times 2^{n-28}} & 6422528&\times 2^{n-28} &9586080.6&\times 2^{n-28}\\
\hline
\end{array}
\]
\end{center}
\caption{A summary of the largest $(1,k)$-overlap free codes obtained by our three constructions. Bold font indicates the best of the three constructions, or a tight upper bound.}
\label{tab:combined}
\end{table*}

Finally, it is perhaps also of interest to compute the exact size of the codes obtained from the the \textsf{Zero Block Construction} for ``small'' values of $k$. We use the formula (\ref{zeroblock.eq}) from Theorem \ref{zeroblock.thm}. For a fixed ``small'' value of $k$, we choose 
$z$ to maximize 
$F_{k+1}^{(z)} \times 2^{-z}$. This is easily done by iterating through the possible values of $z$ to see which one gives the largest result. The exact values $F_{k+1}^{(z)}$ are computed very quickly from the recurrence relation (\ref{fib.eq}).

We present some data in Table \ref{tab3} comparing the \textsf{Zero Block Construction} to the \textsf{$m$-minimum Construction}. For the \textsf{Zero Block Construction}, we also include the optimal value of $z$. Table~\ref{tab:combined} provides a summary of the constructions and bounds in this paper. It is interesting to observe that the \textsf{Zero Block Construction} performs almost as well as the \textsf{$m$-minimum Construction} in all cases, and it gives the same result in many cases. However, the computations of the bounds for the \textsf{Zero Block Construction} are amazingly fast. For example, it is almost instantaneous to compute the lower bound 
\begin{center}
\begin{tabular}{@{}c@{}}
\multicolumn{1}{@{}l}{5745596237141382}\\
\multicolumn{1}{@{\hspace*{3em}}c@{\hspace*{3em}}}{785608786499535716424326}\\
\multicolumn{1}{r@{}}{792561835200479232 $\times 2^{n-200}$}
\end{tabular}
\end{center}

\section{Non-overlapping codes}
\label{classical.sec}

We can apply the techniques of Section \ref{simon.const} to the construction of ``classic'' non-overlapping codes.
Again, we restrict our attention to the binary case.
The following construction is due to Gilbert and Levenshtein; it has been re-discovered several times, and is used in many applications.
See \cite{Chee,Gil60,Lev64,Lev70,Levy}.

\begin{Construction}[\textsf{Gilbert--Levenshtein Construction}]
Suppose $n$ is a given positive integer. For $z=1 , \dots ,  k-1$, we construct a code $L_z$ as follows:
\begin{itemize}
\item each codeword $c = (c_1, \dots , c_n) \in L_z$ begins with a block of $z$ consecutive $0$'s,
\item $c_{z+1} = c_n = 1$, and
\item the sequence $(c_{z+1}, \dots , c_{n-1})$ does not contain $z$ consecutive $0$'s.
\end{itemize}
\end{Construction}
It is clear that $|L_z|$ equals the number of binary sequences of length $n - z - 2$ that do not contain $z$ consecutive $0$'s.
Hence, from Lemma \ref{seq.lem}, we have the following.
\begin{Lemma}
$|L_z| = F_{n-z}^{(z)}$.
\end{Lemma}

Of course we would choose $z$ to maximize $|L_z|$. Let $S(n)$ denote the size of the code obtained from the \textsf{Gilbert--Levenshtein Construction}. The following result is immediate.

\begin{Theorem}
\label{lev.thm}
\begin{equation}
\label{Lev.eq}
S(n) =  \max \left\{ F_{n-z}^{(z)}  : 1 \leq z \leq n-1 \right\} .
\end{equation}
\end{Theorem}

We note that the connection between the \textsf{Gilbert--Levenshtein Construction} and the $n$-step Fibonacci numbers was pointed out by Chee {\it et al.}  \cite{Chee}. In fact, the entries in the third column of \cite[Table 1]{Chee} are computed using the formula (\ref{Lev.eq}).

We can use the techniques developed in  Section  \ref{simon.const} to give an explicit, non-asymptotic lower bound on $S(n)$.

\begin{Lemma}
\[F_{n-z}^{(z)} > (1-n\, 2^{-z-1}) 2^{n-z-2}.\]
\end{Lemma}
\begin{proof}
If we take $k = n - z - 1$ in Lemma \ref{events2.lem}, we obtain
 \[ F_{n-z}^{(z)} \geq (1-(n-z-1)\, 2^{-z-1}) 2^{n-z-2}.\] 
 Clearly,
 \[ 1-(n-z-1)\, 2^{-z-1} > 1-n\, 2^{-z-1},\] so the stated bound follows.
\end{proof}

We now choose $z$ to maximize the function
$h(z) = (1-n\, 2^{-z-1}) 2^{-z-2}$. The maximum occurs when  $z=\log_2 k$, which of course might not be an integer.
Choose $z$ to be an integer in the interval $\left[ \log_2 \frac{3n}{4}, \log_2 \frac{3n}{2} \right].$ 
Then we have \[ h\left(\log_2 \frac{3n}{4}\right) = h\left(\log_2 \frac{3n}{2}\right) = \frac{1}{9n} ,\]
and we obtain the following theorem.

\begin{Theorem}
\label{classic.thm}
$S(n)>(1/9n){2^n}$.
\end{Theorem}

When $n$ is a power of $2$, the maximum value of $h(z)$ occurs when $z = \log_2 n$, and so we do slightly better:
\begin{Theorem}
\label{classic2.thm}
If $n$ is a power of two, then $S(n)>(1/8n){2^n}$.
\end{Theorem}

These bounds improve previous explicit bounds. In Bilotta, Pergola and  Pinzani \cite{BPP12}, an explicit construction based on Dyck paths was given. However, it was observed by Chee {\it et al.} \cite{Chee} that this construction does not yield a lower bound of the form $S(n)>(c/n){2^n}$ for any constant $c > 0$. Also, Blackburn \cite{Bl} proved that $S(n) > (3/64n){2^n}$; our lower bound from Theorem \ref{classic.thm} is stronger. 

As far as asymptotic bounds are concerned,
Levenshtein \cite{Lev64} proved that 
\[ \limsup_{n \rightarrow \infty} S(n) 
\geq (1/2en){2^n} \approx (1/5.436n){2^n}.\]
Levenshtein's asymptotic bound also follows easily from Corollary \ref{cor:noncyclic_zeros} and Theorem~\ref{lev.thm}, as we now demonstrate.

\begin{Theorem}
$\displaystyle\limsup_{n \rightarrow \infty} S(n) 
\geq (1/2en){2^n}.$
\end{Theorem}

\begin{proof}
We prove that
\[ \limsup_{n \rightarrow \infty} \frac{2n S(n) }{2^n} \geq \frac{1}{e}.\]
From Theorem \ref{lev.thm},
we have
\[ \frac{2n S(n) }{2^n} \geq \frac{2n F_{n-z}^{(z)}}{2^n}\] for any $z$ such that $1 \leq z \leq n-1 $.
Let $\ell = 2^a$, $z = a-1$ and $n = \ell + a + 1$ for a positive integer $a$. Then $n-z$ = $\ell + 2$. For these values of $n$ and $z$, we
compute 
\begin{align*} \frac{2n F_{n-z}^{(z)}}{2^n} &= \frac{2(\ell+a+1)}{2^{a+1}}\times \frac{F_{\ell+2}^{(a-1)}}{2^\ell}\\
&= \frac{\ell + a + 1}{\ell} \times \frac{F_{\ell+2}^{(a-1)}}{2^{\ell}}.
\end{align*}
It is clear that $(\ell + a + 1)/{\ell}$ approaches $1$ as $a \rightarrow \infty$, because $\ell = 2^a$.
Also, from Corollary \ref{cor:noncyclic_zeros}, ${F_{\ell+2}^{(a-1)}}/{2^{\ell}}$ approaches $1/e$ as $a \rightarrow \infty$.
The desired result follows. 
\end{proof}

\section{Discussion and Summary}
\label{summary.sec}
In this paper, we have mainly concentrated on $(1,k)$-overlap-free codes over a binary alphabet. 
Our constructions and bounds are actually quite close. 
There are many possible avenues for future research, including studying variable-length analogs, studying codes over non-binary alphabets, or investigating codes with other forbidden overlaps. One direction that might be fruitful for applications is the investigation of codes which are simultaneously $(1,k)$-overlap-free and $(n-k,n-1)$-overlap-free, where $k<\tfrac{n}{2}$.

The \textsf{Zero Block Construction} is inspired by a classical construction of non-overlap\-ping codes due to Gilbert and Levenshtein. It is surprising to us that the \textsf{$m$-minimum Construction} can sometimes yield better codes.
Here is one specific question relating to these two constructions from Section \ref{m-min.sec}: Do the \textsf{$m$-minimum Construction} and \textsf{Zero Block Construction} give the same bound for infinitely many values of $n$?

Finally, we note that the constructions in Section \ref{classical.sec} are most effective when $n$ is close to a power of two. 
We ask if there are constructions that are asymptotically better when $n$ is not of this form, for example when $n=\lfloor 2^{(a+1)/2}\rfloor$  as $a\rightarrow\infty$?


\begin{thebibliography}{10}
\providecommand{\url}[1]{#1}
\csname url@samestyle\endcsname
\providecommand{\newblock}{\relax}
\providecommand{\bibinfo}[2]{#2}
\providecommand{\BIBentrySTDinterwordspacing}{\spaceskip=0pt\relax}
\providecommand{\BIBentryALTinterwordstretchfactor}{4}
\providecommand{\BIBentryALTinterwordspacing}{\spaceskip=\fontdimen2\font plus
\BIBentryALTinterwordstretchfactor\fontdimen3\font minus
  \fontdimen4\font\relax}
\providecommand{\BIBforeignlanguage}[2]{{%
\expandafter\ifx\csname l@#1\endcsname\relax
\typeout{** WARNING: IEEEtranS.bst: No hyphenation pattern has been}%
\typeout{** loaded for the language `#1'. Using the pattern for}%
\typeout{** the default language instead.}%
\else
\language=\csname l@#1\endcsname
\fi
#2}}
\providecommand{\BIBdecl}{\relax}
\BIBdecl

\bibitem{Bar}
E.~Barcucci, A.~Bernini, S.~Bilotta, and R.~Pinzani, ``A 2{D} non-overlapping
  code over a {$q$}-ary alphabet,'' \emph{Cryptogr. Commun.}, vol.~10, no.~4,
  pp. 667--683, 2018.

\bibitem{Bil}
S.~Bilotta, ``Variable-length non-overlapping codes,'' \emph{IEEE Trans.
  Inform. Theory}, vol.~63, no.~10, pp. 6530--6537, 2017.

\bibitem{BPP12}
S.~Bilotta, E.~Pergola, and R.~Pinzani, ``A new approach to cross-bifix-free
  sets,'' \emph{IEEE Trans. Inform. Theory}, vol.~58, no.~6, pp. 4058--4063,
  2012.

\bibitem{Bl}
S.~R. Blackburn, ``Non-overlapping codes,'' \emph{IEEE Trans. Inform. Theory},
  vol.~61, no.~9, pp. 4890--4894, 2015.

\bibitem{Chee}
Y.~M. Chee, H.~M. Kiah, P.~Purkayastha, and C.~Wang, ``Cross-bifix-free codes
  within a constant factor of optimality,'' \emph{IEEE Trans. Inform. Theory},
  vol.~59, no.~7, pp. 4668--4674, 2013.

\bibitem{Cheng}
L.~Cheng, T.~G. Swart, H.~C. Ferreira, and K.~A.~S. Abdel-Ghaffar, ``Codes for
  correcting three or more adjacent deletions or insertions,'' in \emph{Proc.
  IEEE Int. Symp. Inf. Theory (ISIT)}, June 2014, pp. 1246--1250.

\bibitem{Dres}
G.~P.~B. Dresden and Z.~Du, ``A simplified {B}inet formula for
  {$k$}-generalized {F}ibonacci numbers,'' \emph{J. Integer Seq.}, vol.~17, no.
  4, Article 14.4.7, 2014.

\bibitem{Gil60}
E.~N. Gilbert, ``Synchronization of binary messages,'' \emph{IRE Trans. Inform.
  Theory}, pp. 470--477, 1960.

\bibitem{Lev64}
V.~I. Leven\v{s}te\u{\i}n, ``Decoding automata which are invariant with respect
  to their initial state,'' \emph{Probl. Cybern.}, vol.~12, pp. 125--136, 1964.

\bibitem{Lev65}
------, ``Binary codes capable of correcting deletions, insertions and
  reversals,'' \emph{Dokl. Akad. Nauk Tadzhik. SSR}, vol. 163, pp. 845--848,
  1965, (in Russian).

\bibitem{Lev70a}
------, ``Asymptotically optimum binary codes with correction for losses of one
  or two adjacent bits,'' \emph{Syst. Theory Res.}, vol.~19, pp. 298--304,
  1970.

\bibitem{Lev70}
------, ``Maximal number of words in codes without overlap,'' \emph{Probl. Inf.
  Transm.}, vol.~6, pp. 355--357, 1970.

\bibitem{Levy}
M.~Levy and E.~Yaakobi, ``Mutually uncorrelated codes for {DNA} storage,''
  \emph{IEEE Trans. Inform. Theory}, vol.~65, no.~6, pp. 3671--3691, 2019.

\bibitem{Wolf}
\BIBentryALTinterwordspacing
T.~Noe, T.~Piezas~III, and E.~Weisstein, ``Fibonacci $n$-step number,'' from
  \textsl{MathWorld}--A Wolfram Web Resource. [Online]. Available:
  \url{https://mathworld.wolfram.com/Fibonaccin-StepNumber.html}
\BIBentrySTDinterwordspacing

\bibitem{Schoe}
C.~Schoeny, A.~Wachter-Zeh, R.~Gabrys, and E.~Yaakobi, ``Codes correcting a
  burst of deletions or insertions,'' \emph{IEEE Trans. Inform. Theory},
  vol.~63, no.~4, pp. 1971--1985, 2017.

\bibitem{YKM}
S.~M.~H. Tabatabaei~Yazdi, H.~M. Kiah, Gabrys, Ryan, and O.~Milenkovic,
  ``Mutually uncorrelated primers for {DNA}-based data storage,'' \emph{IEEE
  Trans. Inform. Theory}, vol.~64, no.~9, pp. 6283--6296, 2018.

\bibitem{Wang}
G.~Wang and Q.~Wang, ``{$q$}-ary non-overlapping codes: a generating function
  approach,'' \emph{IEEE Trans. Inform. Theory}, vol.~68, no.~8, pp.
  5154--5164, 2022.

\end{thebibliography}

\begin{IEEEbiographynophoto}{Simon R. Blackburn} (M'12, SM'19) was born in Beverley, Yorkshire, England in 1968. He received a BSc in Mathematics from Bristol in 1989, and a DPhil in Mathematics from Oxford in 1992.

He has worked in the Mathematics Department at Royal Holloway University
of London since 1992, and is currently a Professor of Pure Mathematics. His
research interests include algebra, combinatorics and associated applications in
cryptography and communication theory.
\end{IEEEbiographynophoto}
\vskip -2\baselineskip plus -1fil
\begin{IEEEbiographynophoto}{Navid~Nasr~Esfahani} (M'18)
received the B.Sc. degree from the Isfahan University of Technology, Isfahan, Iran, in 2011, the M.Sc. degree from the University of Manitoba, Winnipeg, MB, Canada, in 2014, and the Ph.D. degree from the Cheriton School of Computer Science, University of Waterloo, Waterloo, ON, Canada in 2021. He then continued his research as a Post-Doctoral Fellow at the University of Waterloo. In 2023, he joined the Department of Computer Science at the Memorial University of Newfoundland, Canada, as an Assistant Professor.

His research interests include cryptography, information theory, information theoretic security,
privacy, and combinatorics.
\end{IEEEbiographynophoto}
\vskip -2\baselineskip plus -1fil
\begin{IEEEbiographynophoto}{Donald~L.~Kreher} 
(born in Albany, New York, U.S.A. in 1955) 
obtained a joint computer science and mathematics Ph.D. from 
the University of Nebraska in 1984 and held academic positions at Rochester 
Institute of Technology from 1984 to 1989, 
the University of Wyoming from 1989 to 1991,
and Michigan Technological University from 1991 to 2020 when he retired 
as an emeritus professor.

In 1995, Professor Kreher was awarded the Marshall Hall Medal, awarded by the 
Institute of Combinatorics and its Applications.

His research interests include computational and algebraic methods for 
determining the structure and existence of combinatorial configurations, such 
as designs, graphs, error-correcting codes, cryptographic systems and extremal 
set systems. 
\end{IEEEbiographynophoto}
\vskip -2\baselineskip plus -1fil
\begin{IEEEbiographynophoto}{Douglas~R.~Stinson}
(born in 1956 in Guelph, Ontario) is a Canadian mathematician and cryptographer, currently Professor Emeritus at the University of Waterloo. Stinson received his B.Math from the University of Waterloo in 1978, his M.Sc. from Ohio State University in 1980, and his Ph.D. from the University of Waterloo in 1981. He was at the University of Manitoba from 1981 to 1989 and the University of Nebraska-Lincoln from 1990 to 1998. Since 1998 he has been at the University of Waterloo, retiring in 2019. Professor Stinson was awarded the 1994 Hall Medal and the 2022 Stanton Medal by the Institute of Combinatorics and its Applications. In 2011, he was named as a Fellow of the Royal Society of Canada. His research interests include combinatorics, cryptography, algorithms and information security.
\end{IEEEbiographynophoto}
\vfill

\end{document}